\documentclass{aamas2011_cameraReady}

\newtheorem{lemma}{Lemma}
\newtheorem{theorem}{Theorem}
\newtheorem{proposition}{Proposition}
\newtheorem{definition}{Definition}

\begin{document}

\AuthorsForCitationInfo{Arthur Carvalho and Kate Larson}

\TitleForCitationInfo{A Truth Serum for Sharing Rewards}

\title{A Truth Serum for Sharing Rewards}

\numberofauthors{2}

\author{
\alignauthor
Arthur Carvalho\\
       \affaddr{Cheriton School of Computer Science}\\
       \affaddr{University of Waterloo}\\
       \affaddr{Waterloo, Ontario, Canada}\\
       \email{a3carval@cs.uwaterloo.ca}
\alignauthor 
Kate Larson\\
       \affaddr{Cheriton School of Computer Science}\\
       \affaddr{University of Waterloo}\\
       \affaddr{Waterloo, Ontario, Canada}\\
       \email{klarson@cs.uwaterloo.ca}
}

\maketitle

\begin{abstract}
We study a problem where a group of agents has to decide how a joint reward should be shared among them. We focus on settings where the share that each agent receives depends on the \emph{subjective opinions} of its peers concerning that agent's contribution to the group. To this end, we introduce a mechanism to elicit and aggregate subjective opinions as well as for determining agents' shares. The intuition behind the proposed mechanism is that each agent who believes that the others are telling the truth has its expected share maximized to the extent that it is well-evaluated by its peers and that it is truthfully reporting its opinions. Under the assumptions that agents are Bayesian decision-makers and that the underlying population is sufficiently large, we show that our mechanism is incentive-compatible, budget-balanced, and tractable. We also present strategies to make this mechanism individually rational and fair.
\end{abstract}

\category{\\J.4}{Social and Behavioral Sciences}{Economics}



\section{Introduction}

Understanding how agents can work together in order to achieve some common goal is a central research topic in the field of multiagent systems~\cite{mas-book}. Questions that are typically analyzed include how and which groups of agents should form~\cite{Rahwan09:Anytime}, how agents should coordinate their actions once they have agreed to work together~\cite{Grosz:collaborativeplans}, how to ensure that the group, once formed, does not disintegrate~\cite{Conitzer06:Complexity}, and how any joint rewards should be divided among the group members~\cite{Moulin:Fair_division}. It is this last question that we address in this paper.

Commonly called \emph{fair division}, the problem of dividing one or several goods among a set of agents, in a way that satisfies a suitable fairness criterion, has been studied in several literatures. In economics, the collective welfare approach is arguably the most influential application of the economic analysis to fair division. It uses the concepts of collective utility functions, in its cardinal interpretation, and social welfare orderings, in its ordinal interpretation, for deciding what makes a reasonable division~\cite{Moulin:Fair_division}. In computer science and, more specifically, artificial intelligence, the fair division problem is traditionally studied in settings where the underlying agents not only have preferences over alternative allocations of goods, but also actively participate in computing an allocation~\cite{mara-survey}.

In this work, we propose a novel game-theoretic model for sharing a joint, homogeneous reward based on the idea of \emph{subjective opinions}. In detail, we consider scenarios where a group has been formed and has accomplished a task for which it is granted a reward, which must be shared among the group members. After observing the individual contributions of the peers in accomplishing the task, each agent is asked to \emph{evaluate} the others. Agents also provide \emph{predictions} about how their peers are evaluated. Thus, we consider two kinds of subjective opinions when sharing the joint reward: \emph{evaluations} and \emph{predictions}. These opinions are elicited and aggregated by a central, trusted entity called the \emph{mechanism}, which is also responsible for sharing the reward based exclusively on the received opinions.

The share received by each agent from the proposed mechanism has two major components. The first one reflects the evaluations received by that agent. The second one is a truth-telling score used to encourage agents to truthfully report their opinions. For computing such scores, the mechanism uses the Bayesian truth serum method~\cite{serum}. The intuition behind the proposed mechanism is that each agent who believes that the others are telling the truth has its expected share maximized to the extent that it is well-evaluated and that it is also telling the truth. Under the assumptions that agents are Bayesian decision-makers and that the underlying population is sufficiently large, we show that our mechanism is incentive-compatible, budget-balanced, and tractable. We also present strategies to make this mechanism individually rational and fair.

Besides this introductory section, the rest of this paper is organized as follows. In Section 2, we describe the model, concepts used throughout the paper, and properties that we wish our mechanism to exhibit. In Section 3, we introduce our mechanism and prove that it satisfies interesting properties. In Section 4, we empirically investigate the influence of the model and mechanism's parameters on agents' shares. In Section 5, we review the literature related to our work. Finally, we conclude in Section 6.

\section{Model and Background}

A set of agents $N = \{1, \dots, n\}$, for $n \geq 3$, has accomplished a task for which it is granted a \emph{reward} $V \in \Re^+$. Every agent is assumed to want more of the reward. Therefore, we can identify an agent's share with its welfare. We are interested in settings where the share of $V$ that an agent receives depends on the \emph{subjective opinions} of its peers concerning that agent's contribution to the group.

We model the private information of an agent as $n-1$\break private signals that the agent receives from its peers. These signals are direct assessments of the peers' performance in accomplishing the joint task, and we call them \emph{truthful evaluations}. Formally, given a positive integer parameter $M$, for $1 \leq M \leq V$, the signals observed by agent $i$ are represented by the vector $\mathbf{t}_i = (t_{i}^1,\dots,t_{i}^{i-1},t_{i}^{i+1},\dots, t_{i}^n)$, where $t_i^j \in \{ 1, \dots, M \}$ represents the signal observed by agent $i$ coming from agent $j$. Thus, $\mathbf{t}_i$ is the vector with the truthful evaluations made by agent $i$ regarding the contributions of its peers in accomplishing the task. In this way, the parameter $M$ represents the top possible evaluation that an agent can give or receive, and we assume that its value is common knowledge. For each agent $j \in N$, let $\omega_j \in \Delta^M$ (unit simplex in $\Re^{M}$) be an unknown parameter representing the distribution of the truthful evaluations for agent $j$.

Based on their truthful evaluations, agents can make predictions about how their peers are evaluated. The predictions made by agent $i$ are formally represented by the vector $\mathbf{r}_i = (r_{i}^1, \dots, r_{i}^{i-1}, r_{i}^{i+1}, \dots, r_{i}^n)$, where agent $i$'s prediction about the empirical distribution of evaluations received by agent $j$ is $r_i^j = (r_{i}^{j^1}, \dots, r_{i}^{j^M}) \in \Delta^M$, \textit{i.e.}, $ 0 \leq r_i^{j^k} \leq 1$ and $\sum_{k=1}^M r_i^{j^k} = 1$. Mathematically,  $r_{i}^j$ is the expected distribution of truthful evaluations for agent $j$ given agent $i$'s truthful evaluation, $i.e.$, $r_{i}^j = \mathbb{E}[\omega_j | t_i^j]$.

To avoid a biased self-judgment, agents are neither asked to make self-evaluations nor asked to make predictions about their received evaluations. They are requested to report their subjective opinions, namely, evaluations and predictions. We make the following assumptions in our model:

\begin{enumerate}

     \item \emph{Self-interestedness.} Agents act to maximize their expected shares.

	  \item \emph{Common prior.} $\forall j \in N$, there exists a common prior distribution, $p(\omega_j)$, over $\omega_j$.

     \item \emph{Rationality.} Every agent $i$, with truthful evaluation $t_i^j$, forms a posterior by applying Bayes' rule to the common prior $p(\omega_j)$, \textit{i.e.}, $p(\omega_j | t_{i}^j)$.

      \item \emph{Stochastic relevance.} $\forall i, q, j \in N, \, p(\omega_j | t_{i}^j) = p(\omega_j|t_{q}^j)$ if and only if $t_{i}^j = t_{q}^j$.

      \item \emph{Large population.} The population of agents must be sufficiently large so that a single evaluation for an agent cannot significantly affect the empirical distribution of evaluations received by that agent.

				\item \emph{Independent signals.} The signals observed by an agent are independent of each other. Formally, given $i,j,k \in N$, and $x, y \in \{1,\dots, M\}, p(t_i^j = x \,|\, t_i^k = y) = p(t_i^j = x)$.
\end{enumerate}

The first assumption means that agents are \emph{risk neutral} \cite{mascolell-whinston-green}. The second assumption means that  agents have common prior distributions over the distributions of the truthful evaluations for their peers. The third assumption means that these priors are consistent with Bayesian updating. These first three assumptions are traditional in both game theory~\cite{osborne-rubinstein} and multiagent systems \cite{mas-book} literature, and they essentially mean that agents are Bayesian decision-makers. The fourth assumption means that different truthful evaluations imply different posterior distributions, and vice-versa. By far, the most stringent assumption is the requirement of a large population. Later in this paper, we discuss the implications of such assumption and how to circumvent it. Finally, the last assumption implies that the truthful evaluation of an agent for a peer does not influence that agent's truthful evaluation for other peer.

A consequence of self-interest is that agents may deliberately lie when reporting their evaluations and/or predictions. For example, an agent may intentionally give all other agents a low evaluation so that, in comparison, it looks good and receives a greater share of $V$.  Therefore, we distinguish between the truthful evaluations made by each agent $i \in N$, $\mathbf{t}_i$, and the evaluations that agent $i$ reports, $\mathbf{x}_i=(x_{i}^1, \dots, x_{i}^{i-1}, x_{i}^{i+1}, \dots, x_{i}^n)$. Similarly, we distinguish between the truthful predictions made by each agent $i \in N$, $\mathbf{r}_i$, and the predictions that agent $i$ reports, $\mathbf{y}_i=(y_{i}^1, \dots, y_{i}^{i-1}, y_{i}^{i+1}, \dots, y_{i}^n)$.

We define the \emph{strategy} of agent $i$, $\mathbf{s}_i = (\mathbf{x}_i, \mathbf{y}_i)$, to be its reported opinions. $S_i$ is the set of strategies available to agent $i$, and $S = S_1 \times \ldots \times S_n$. We note that the parameter $M$ fully determines the strategies available to the agents. Each vector $\mathbf{s} = (\mathbf{s}_1, \dots, \mathbf{s}_n) \in S$ is a \emph{strategy profile}. As customary, let the subscript ``$-i$" denote a vector without agent $i$'s component, \textit{e.g.}, $\mathbf{s}_{-i}=(\mathbf{s}_1, \dots, \mathbf{s}_{i-1}, \mathbf{s}_{i+1}, \dots, \mathbf{s}_n)$. If the opinions reported by agent $i$ are equal to its truthful opinions, \textit{i.e.}, $\mathbf{x}_i=\mathbf{t}_i$ and $\mathbf{y}_i=\mathbf{r}_i$, then we say that agent $i$'s strategy is \emph{truthful}.

Opinions are elicited and aggregated by a central, trusted entity called the \emph{mechanism}, which is also responsible for sharing the reward among the agents. This entity relies only on the reported opinions when determining agents' shares, and so it has no additional information. Formally:

\begin{definition}[Mechanism]
A mechanism is a sharing function, $\Gamma:S \rightarrow \Re^n$, which maps each strategy profile to a vector of shares.
\end{definition}

We denote the share of $V$ given to agent $i$, when all the reported opinions are $\mathbf{s}$, by $\Gamma_i(\mathbf{s})$. We use $\Gamma_i$ when $\mathbf{s}$ is either irrelevant or clear from the context. Throughout this paper, we use the solution concept called \emph{Bayes-Nash equilibrium}.

\begin{definition}[Bayes-Nash equilibrium] We say \newline that the strategy profile $\mathbf{s} = (\mathbf{s}_1, \dots, \mathbf{s}_n)$ is a Bayes-Nash equilibrium if for each agent $ i$, and strategy $\mathbf{s}^{\prime}_i \neq \mathbf{s}_i \in S_i$, $\mathbb{E} \left[ \Gamma_i(\mathbf{s}_i,\mathbf{s}_{-i}) | \mathbf{t}_i,\mathbf{r}_i \right] \geq$  $\mathbb{E}\left[ \Gamma_i(\mathbf{s}^{\prime}_i,\mathbf{s}_{-i}) | \mathbf{t}_i,\mathbf{r}_i\right]$.
\end{definition}

In words, for each agent $i \in N$, $\mathbf{s}_i$ is the best response, in an expected sense, that agent $i$ has to $\mathbf{s}_{-i}$ given its truthful opinions $(\mathbf{t}_i, \mathbf{r}_i)$. The expectation in taken with respect to the posterior distributions. When the inequality in Definition 2 holds strictly (with \textquotedblleft$>$\textquotedblright \space instead of \textquotedblleft$\geq$\textquotedblright), then the strategy profile $\mathbf{s}$ is called a \emph{strict Bayes-Nash equilibrium}.

\subsection{Properties}

There are several key properties we wish mechanisms to have. We introduce them in this subsection.

\begin{definition}[Fairness]
Consider a strategy profile $\mathbf{s} \in S$ in which the reported evaluation of every agent $z$ for agent $i$ is paired up with agent $z$'s reported evaluation  for agent $j$, for $i\neq j\neq z \in N$, so that $x_z^i > x_z^j $. Further, the evaluations of agent $i$ and agent $j$ for each other are paired up, so that $x_j^i > x_i^j $. Then, we say that a mechanism is fair if  $ \Gamma_i(\mathbf{s}) > \Gamma_j(\mathbf{s})$.
\end{definition}

In words, if an agent unanimously receives better evaluations than a peer, then that agent should also receive a greater share of the joint reward than its peer.

\begin{definition}[Budget Balance]
A mechanism is \newline budget-balanced if $\forall \mathbf{s} \in S, \sum_{i=1}^n \Gamma_i(\mathbf{s}) = V$.
\end{definition}

In words, a budget-balanced mechanism allocates the entire reward $V$ back to the agents. As stated, this is a strong definition because we do not put constraints on $\mathbf{s}$, \textit{e.g.}, we do not require $\mathbf{s}$ to be an equilibrium strategy profile.

\begin{definition}[Individual Rationality]
A mecha-\newline nism is individually rational if $\forall i\in N, \forall\mathbf{s} \in S, \Gamma_i(\textbf{s}) \geq 0$.
\end{definition}

This condition requires the share received by each agent to be greater than or equal to zero. In other words, all agents are weakly better off participating in the mechanism than not participating at all.

\begin{definition}[Incentive Compatibility] A mechanism is incentive-compatible if collective truth-telling is an equilibrium strategy profile. 
\end{definition}

Since we are working with Bayes-Nash equilibrium, an incentive-compatible mechanism implies that it is best, in an expected sense, for each agent to tell the truth provided that the others are also doing so.

\begin{definition}[Tractability] A mechanism is trac-table if it computes agents' shares in polynomial time.
\end{definition}

By no means do we argue that the properties defined in this section are exhaustive. However, we believe that they are among the most desirable ones in practical applications.

\subsection{The Bayesian Truth Serum Method}

Prelec~\cite{serum} proposes an incentive-compatible scoring meth-od, called the \emph{Bayesian Truth Serum} (BTS), which works on a single multiple-choice question with a finite number of alternatives. Each responder is requested to endorse the answer mostly likely to be true and to predict the empirical distribution of the endorsed answers.

Responders are evaluated by the accuracy of their predictions (how well they matched the empirical frequency) as well as how \textit{surprisingly common} their answers are. For example, an answer endorsed by 50\% of the population against a predicted frequency of 25\% is surprisingly common. The responders who endorsed that answer should receive a high score. If predictions averaged 75\%, an answer endorsed by 50\% of the population would be \textit{surprisingly uncommon} and, consequently, the responders who endorsed it would receive a lower score. The surprisingly common criterion exploits the \emph{false consensus effect} to promote truthfulness, \textit{i.e.}, the general tendency of responders to overestimate the degree of agreement that the others have with them~\cite{false-consensus}.

In our work, the BTS method is used exclusively as a tool to promote truthfulness. This method is very convenient because it does not require objective answers to score opinions, \textit{i.e.}, it is possible to work with subjective information, where an absolute truth is practically unknowable, and still be able to reward truthfulness. Questions that are considered in our work have the form: \textquotedblleft What is the evaluation deserved by agent $j$?\textquotedblright, where the possible answers are values inside the set $\{1, \dots, M\}$. For illustration purpose, consider a question asking for the evaluation deserved by agent $j$. Using the notation previously defined, let  $h(x_i^j, k)$ be a zero-one indicator function, \textit{i.e.},

\begin{displaymath}
h(x_i^j, k) = \left\{ \begin{array}{ll}
1 & \textrm{if $x_i^j = k$},\\
0 & \textrm{otherwise}.\\
\end{array} \right.
\end{displaymath}

The score returned by the BTS method to agent $i$, given its reported evaluation $x_i^j$ and prediction $y_i^j$, is calculated as follows:

\begin{equation}
\label{BTS}
\mathbb{R}(i,j) = \sum_{k = 1}^M h(x_i^{j}, k) \ln \frac{\bar{x}_k}{\bar{y}_k} + \sum_{k = 1}^M \bar{x}_k \ln  \frac{ (1-\epsilon)y_i^{j^k} + \frac{\epsilon}{M}}  {\bar{x}_k},
\end{equation}

\noindent where $\bar{x}_k$ is the average frequency of evaluation $k$, and $\bar{y}_k$ is the geometric average of the predicted frequencies of evaluation $k$:

\begin{eqnarray*}
\bar{x}_k &=& (1-\epsilon)\left(\frac{1}{n-1}  \sum_{q \neq j}  h(x_q^j, k)\right) + \frac{\epsilon}{M} \\
\bar{y}_k &=& \exp\left(\frac{1}{n-1}  \sum_{q \neq j} \ln \left((1-\epsilon)y_q^{j^k} + \frac{\epsilon}{M} \right)\right)
\end{eqnarray*}

\noindent and $\epsilon$, for $0 < \epsilon < 1$, is a recalibration coefficient to adjust predictions and averages away from $0/1$ extreme values.

The BTS method has two major components. The first one, called the \emph{information score}, evaluates the evaluation given by agent $i$ to agent $j$ according to the log-ratio of its actual-to-predicted endorsement frequencies. An evaluation scores high to the extent that it is more common than collectively predicted. The second component, called the \emph{prediction score}, is a penalty proportional to the relative entropy between the empirical distribution of evaluations for agent $j$ and agent $i$'s prediction of that distribution. For a small $\epsilon$, the best prediction score is attained when a reported prediction matches the empirical distribution of evaluations.

It is interesting to note that Equation~\ref{BTS} is slightly different from the original BTS method, which uses $\epsilon = 0$. By using a small recalibration coefficient, we can avoid problems related to values that are not well-defined, \textit{e.g.}, $\ln(0)$ and $\ln(0/0)$. Any distortion in incentives can be made arbitrarily small by making $\epsilon$ sufficiently small. Under the assumptions made in the beginning of this section, and using Equation~\ref{BTS} to compute agents' scores, the following theorems hold~\cite{serum}:

\begin{theorem}
Collective truth-telling is a strict Bayes-Nash equilibrium.
\end{theorem}

\begin{theorem}
The BTS method is zero-sum.
\end{theorem}

Theorem 1 means that the strict best response of an agent, in an expected sense, when everyone else is telling the truth is also to tell the truth. Theorem 2 means that the sum of the scores received by the agents is equal to zero, \textit{i.e.}, $\sum_{i\neq j}\mathbb{R}(i,j) = 0$. In what follows, we provide bounds for the scores returned by the BTS method.

\begin{lemma}
$\forall i \neq j,\, \mathbb{R}(i,j) \in \left[-2\ln(\frac{M}{\epsilon}),\, \ln(\frac{M}{\epsilon})\right]$.
\end{lemma}

\begin{proof}
We start by noting that:

\begin{displaymath}
0 < \frac{\epsilon}{M} \leq  \bar{x}_k,  \bar{y}_k  \leq 1-\epsilon + \frac{\epsilon}{M} < 1.
\end{displaymath}

Focusing first on the lower-bound, we analyze each part of Equation~\ref{BTS} separately. Starting with the information score, we have:

\begin{eqnarray}
\label{bound-bts-1}
\sum_{k = 1}^M h(x_i^{j}, k) \ln \frac{\bar{x}_k}{\bar{y}_k} 
&\geq& 
\sum_{k = 1}^M h(x_i^{j}, k) \ln \bar{x}_k \\ \nonumber
&\geq& 
\ln\frac{\epsilon}{M},  \nonumber
\end{eqnarray}

\noindent where the inequalities follow, respectively, from the facts that $0 < \bar{y}_k < 1$, and $\frac{\epsilon}{M} \leq  \bar{x}_k  < 1$. Moving to the prediction score, we have:

\begin{eqnarray}
\label{bound-bts-2}
\sum_{k = 1}^M \bar{x}_k \ln \frac{ (1-\epsilon)y_i^{j^k} + \frac{\epsilon}{M}}  {\bar{x}_k} 
&\geq& 
\sum_{k = 1}^M \bar{x}_k \ln \frac{\epsilon}{M} \\\nonumber
&=& 
\left(1 - \epsilon + \frac{\epsilon}{M}\right) \ln\frac{\epsilon}{M} \\\nonumber
&\geq& \ln\frac{\epsilon}{M},
\end{eqnarray}

\noindent where the first inequality follows from the facts that $0 < \bar{x}_k < 1$ and $(1-\epsilon)y_i^{j^k} \geq 0$. The second inequality follows from the facts that $\ln(\epsilon/M) < 0$, and $0 < \left(1 - \epsilon + \frac{\epsilon}{M}\right) < 1$. Joining (\ref{bound-bts-1}) and (\ref{bound-bts-2}), we have:

\begin{eqnarray*}
\mathbb{R}(i,j) &\geq& 2\ln\frac{\epsilon}{M} \\
&=& -2\ln\frac{M}{\epsilon}.
\end{eqnarray*}

Focusing now on the upper-bound of Equation~\ref{BTS}, we start by analyzing the information score:

\begin{eqnarray}
\sum_{k = 1}^M h(x_i^{j}, k) \ln \frac{\bar{x}_k}{\bar{y}_k} 
&\leq&
\sum_{k = 1}^M h(x_i^{j}, k) \ln \frac{1}{\bar{y}_k} \\ \nonumber
&\leq& 
\ln\frac{1}{\frac{\epsilon}{ M}} \\ \nonumber
&=&
\ln\frac{M}{\epsilon}, \\ \nonumber
\end{eqnarray}

The inequalities follow from the fact that $\frac{\epsilon}{M} \leq \bar{x}_k,\bar{y}_k < 1$. Moving to the prediction score, we note that its value is always less than or equal to zero, because it can be seen as the negative of the Kullback-Leibler divergence, which is always greater than or equal to zero~\cite{information_theory}. Thus, we have:

\begin{displaymath}
\mathbb{R}(i,j) \leq \ln\frac{M}{\epsilon}.
\end{displaymath}
\end{proof}

\section{The Mechanism}

In this section, we propose a mechanism for sharing rewards based on subjective opinions. It starts by requesting both evaluations and predictions from the agents. For each vector with evaluations, $\mathbf{x}_i$, the mechanism creates another vector, $\mathbf{\chi}_i = (\chi_i^1,\dots, \chi_i^{i-1}, \chi_i^{i+1}, \dots,\chi_i^n)$, by scaling the elements of $\mathbf{x}_i$ so that they sum up to $V$. Mathematically,

\begin{equation}
\label{scaled-eval}
\forall i,j,\, \chi_i^j = x_i^j\left(\frac{V}{\sum_{q\neq i} x_i^q}\right).
\end{equation}

This simple pre-processing step ensures that the sum of the resulting shares is not orders of magnitude lower than the reward $V$. The share received by each agent $i \in N$ from the mechanism has two major components. The first one, $\bar{\chi}^i$, reflects agent $i$'s received evaluations. It is calculated by summing the scaled evaluations received by agent $i$, and dividing the sum by $n$, \textit{i.e.},

\begin{equation}
\label{scaled-eval2}
\bar{\chi}^i = \frac{\sum_{j\neq i} \chi_j^i}{n}.
\end{equation}

This simple idea of aggregating the scaled evaluations for an agent by summing them and dividing by $n$ helps to ensure important properties for the mechanism. The second component of agent $i$'s share is a truth-telling score. The intuition behind such scores is that agents who believe that the others are telling the truth maximize their expected scores by also telling the truth. The score of agent $i$, $\zeta_i$, is calculated as follows:

\begin{equation}
\label{truth-telling-scores}
\zeta_i = \frac{\sum_{j\neq i} \mathbb{R}(i,j)}{n-1},
\end{equation}

\noindent where $\mathbb{R}(i,j)$ is defined in Equation~\ref{BTS}. Agent $i$'s score is then the arithmetic mean of results returned by the Bayesian truth serum method, where each result is directly related to an evaluation and a prediction reported by agent $i$. Finally, the share of agent $i$ is a linear combination of $\bar{\chi}^i$ and $\zeta_i$, \textit{i.e.},

\begin{equation}
\label{share-bts}
\Gamma_i  = \bar{\chi}^i + \alpha\, \zeta_i,
\end{equation}

\noindent where the constant $\alpha$, for $\alpha > 0$, fine-tunes the weight given to the truth-telling score $\zeta_i$. Its value has an important role in ensuring desirable properties for the mechanism.

The intuition behind the proposed mechanism is that agents who believe that the others are truthfully reporting have their expected shares maximized to the extent that they are well-evaluated and that they are also telling the truth. It is interesting to note that despite the assumptions of prior and posterior distributions, they are neither known nor requested by the mechanism, only evaluations and predictions are elicited from agents.

\subsection{Numerical Example}

A numerical example may clarify the mechanics of the proposed mechanism. Consider six agents indexed by the letters $A, B, C, D, E, F,$ a joint reward $V = 1000$, and assume that $M = 2$. The reported predictions and evaluations can be seen, respectively, in Table~\ref{tab:pred} and Table~\ref{tab:eval}.

\begin{table*}
\centering
\caption{\label{tab:pred} Numerical example: reported predictions.}
\begin{tabular}{|c|c|c|c|c|c|c|c|c|c|c|c|c|} \hline
& \multicolumn{2}{|c|}{\textbf{A}}
& \multicolumn{2}{|c|}{\textbf{B}} 
& \multicolumn{2}{|c|}{\textbf{C}} 
& \multicolumn{2}{|c|}{\textbf{D}} 
& \multicolumn{2}{|c|}{\textbf{E}} 
& \multicolumn{2}{|c|}{\textbf{F}}\\ \hline 
  & \textbf{``1"} & \textbf{``2"} & \textbf{``1"} & \textbf{``2"} & \textbf{``1"} & \textbf{``2"} & \textbf{``1"} & \textbf{``2"} & \textbf{``1"} & \textbf{``2"} & \textbf{``1"} & \textbf{``2"} \\ \hline  
\textbf{A} & - & - & 0 & 1 & 0.4 & 0.6 & 0.2 & 0.8 & 1 & 0 & 0.2 & 0.8 \\ \hline 
\textbf{B} & \emph{0.8} & 0.2 & - & - & 0.2 & 0.8 & 0.2 & 0.8 & 1 & 0 & 0.4 & 0.6  \\ \hline 
\textbf{C} & 0.8 & 0.2 & 0 & 1 & - & - & 0.4 & 0.6 & 1 & 0 & 0.4 & 0.6 \\ \hline 
\textbf{D} & 0.8 & 0.2 & 0.2 & 0.8 & 0.6 & 0.4 & - & - & 0.8 & 0.2 & 0.4 & 0.6 \\ \hline 
\textbf{E} & 0.8 & 0.2 & 0 & 1 & 0.6 & 0.4 & 0.4 & 0.6 & - & - & 0.4 & 0.6 \\ \hline 
\textbf{F} & 0.8 & 0.2 & 0.8 & 0.2 & 0.6 & 0.4 & 0.4 & 0.6 & 0.8 & 0.2 & - & - \\ \hline 
\end{tabular}
\end{table*}

In Table~\ref{tab:pred}, each numeric cell can be interpreted as the prediction made by the agent in the row about the percentage of agents that give the evaluation in the second row of the cell's column (\textbf{``1"} or \textbf{``2"}) to the agent in the first row of the cell's column. For example, the emphasized number $0.8$ means that agent $B$ predicts that 80\% of the population gives the evaluation $1$ to agent $A$.

In Table~\ref{tab:eval}, each numeric cell can be interpreted as the evaluation given by the agent in the row to the agent in the column. For example, the emphasized number $2$ represents $x_A^B$, \textit{i.e.}, the evaluation given by agent $A$ to agent $B$.

Using these evaluations and predictions, and the parameters $\alpha = 100$ and $\epsilon = 0.01$, the  mechanism returns the shares shown in the last column of Table~\ref{tab:shares}. The major components of these shares are shown in the first columns. For illustration's sake, consider the share received by agent $F$. To compute the first component of $\Gamma_F$, the mechanism aggregates the scaled evaluations received by agent $F$ (Equation~\ref{scaled-eval2}):

\begin{eqnarray*}
\bar{\chi}^F &=& \frac{142.86 + 250.00 + 285.71 + 250.00 + 222.22}{6} \\
			    &\approx& 191.80.
\end{eqnarray*}

The second component of $\Gamma_F$ is the arithmetic mean of results returned by the BTS method, where each result is directly related to an evaluation and a prediction submitted by agent $F$ (Equation~\ref{truth-telling-scores}):

\begin{eqnarray*}
\zeta_F &=& \frac{\mathbb{R}(F,A) + \mathbb{R}(F,B) + \mathbb{R}(F,C) + \mathbb{R}(F,D) + \mathbb{R}(F,E)}{5} \\
	 &\approx& \frac{0.58 -1.19 -0.18 -0.11 -0.11}{5}\\
			    &\approx& -0.21.
\end{eqnarray*}

Finally, the share received by agent $F$ from the mechanism is a linear combination of $\bar{\chi}^F$ and $\zeta_F$:

\begin{eqnarray*}
\Gamma_F &=& \bar{\chi}^F + \alpha\,\zeta_F \\
	       &=& 191.80 + 100\times(-0.21) \\
			 &=& 170.80.
\end{eqnarray*}

\begin{table}[h]
\centering
\caption{\label{tab:eval} Numerical example: reported evaluations.}
\begin{tabular}{|c|c|c|c|c|c|c|} \hline
  & \textbf{A} & \textbf{B} & \textbf{C} & \textbf{D} & \textbf{E} & \textbf{F} \\\hline
\textbf{A} & - & \emph{2} & 2 & 1 & 1 & 1 \\ \hline
\textbf{B} & 1 & - & 2 & 2 & 1 & 2 \\ \hline
\textbf{C} & 1 & 2 & - & 1 & 1 & 2 \\ \hline
\textbf{D} & 1 & 2 & 2 & - & 1 & 2 \\ \hline
\textbf{E} & 2 & 2 & 1 & 2 & - & 2 \\ \hline
\textbf{F} & 2 & 2 & 1 & 2 & 1 & - \\ \hline
\end{tabular}
\end{table}

\begin{table}[h]
\centering
\caption{\label{tab:shares} Numerical example: resulting shares.}
\begin{tabular}{|c|c|c|c|} \hline
& $\bar{\chi}^i$ & $\zeta_i$ & $\Gamma_i$ \\ \hline
\textbf{A} & 144.18 &  0.05  & 149.18     \\ \hline
\textbf{B} & 215.61 & -0.06  & 209.61     \\ \hline
\textbf{C} & 170.30 &  0.09  & 179.30     \\ \hline
\textbf{D} & 167.99 & -0.02  & 165.99     \\ \hline
\textbf{E} & 110.12 &  0.15  & 125.12     \\ \hline
\textbf{F} & 191.80 & -0.21  & 170.80     \\ \hline
\end{tabular}
\end{table}

\subsection{Properties}

In this subsection, we show that the proposed mechanism satisfies important properties.

\begin{proposition}
The mechanism is budget-balanced.
\end{proposition}

\begin{proof}

The sum of the shares received by the agents is equal to:

\begin{eqnarray*}
\sum_{ i = 1}^n \left( \bar{\chi}^i + \alpha\,\zeta_i \right) 
&=& \sum_{ i = 1}^n \bar{\chi}^i + \alpha\,\sum_{ i = 1}^n\zeta_i 
\\
&=&\sum_{ i = 1}^n \frac{\sum_{j\neq i} \chi_j^i}{n} +\alpha\,\sum_{ i = 1}^n\frac{\sum_{j\neq i} \mathbb{R}(i,j)}{n-1}
\\
&=& \sum_{ j = 1}^n\frac{\sum_{i\neq j} \chi_j^i}{n} + \alpha\,\sum_{ j = 1}^n\frac{\sum_{i\neq j} \mathbb{R}(i,j)}{n-1}
\\
&=& n\left(\frac{V}{n}\right) + \frac{\alpha}{n-1}\left(\sum_{ j = 1}^n \sum_{i\neq j} \mathbb{R}(i,j)\right).
\end{eqnarray*}

The last equality follows from the fact that the scaled evaluations sum up to $V$ (Equation~\ref{scaled-eval}). From Theorem 2, we know that $\sum_{ i \neq j } \mathbb{R}(i,j) = 0$, thus completing the proof.
\end{proof}

\begin{proposition}
The mechanism is incentive-compatible.
\end{proposition}

\begin{proof}[(Sketch)]

Due to space limitations, we only provide a sketch of the proof. Suppose that every peer of an agent $i \in N$ is truthfully reporting its opinions. We prove that the strict best response for agent $i$, in an expected sense, is also to tell the truth. We start by observing that the share received by agent $i$ (Equation~\ref{share-bts}) can be written as $c_1 + c_2\sum_{j \neq i} \mathbb{R}(i,j)$, where $c_1$ and $c_2$ are positive constants, from agent $i$'s point of view, because they do not depend on the opinions reported by agent $i$. Due to the assumption of independent signals (Assumption 6, Section 2), we can restrict ourselves to find the strategy of agent $i$ that maximizes $\mathbb{E}\left[c_1 + c_2\mathbb{R}(i,j)\right] = c_1 + c_2\mathbb{E}\left[\mathbb{R}(i,j)\right] $, which in turn is strictly maximized when agent $i$ tells the truth (Theorem 1). Thus, the mechanism is incentive-compatible.
\end{proof}

\begin{proposition}
The mechanism is tractable.
\end{proposition}

\begin{proof}
The pre-processing step (Equation~\ref{scaled-eval}) is computed in $O(n^2)$. Thereafter, for each agent $i \in N$, Equation~\ref{scaled-eval2} is computed in $O(n)$, Equation~\ref{truth-telling-scores} is computed in $O(n^2M)$ (since Equation~\ref{BTS} can be computed in $O(nM)$), and Equation~\ref{share-bts} is computed in $O(1)$. Thus, the mechanism runs in $O(n^3M)$ time. 
\end{proof}

\begin{proposition}
If $M \leq \sqrt{n-2}$ and $\alpha \leq \frac{V}{3Mn^2\ln\left(\frac{M}{\epsilon}\right)}$, then the mechanism is fair.
\end{proposition}

\begin{proof}

Consider a pair of agents $i, j \in N$ and a strategy profile $\mathbf{s} \in S$ where $x_j^i > x_i^j$ and, for every other agent $z \neq i,j$, $x_z^i > x_z^j$. For the mechanism to be considered fair, its resulting shares must satisfy the following inequality:

\begin{eqnarray}
\label{proof-fair}
\Gamma_i(\mathbf{s}) > \Gamma_j(\mathbf{s}) &\equiv&
\bar{\chi}^i + \alpha\, \zeta_i >
\bar{\chi}^j + \alpha\, \zeta_j \nonumber\\
&\equiv& \alpha < \frac{\bar{\chi}^i - \bar{\chi}^j}{\zeta_j - \zeta_i}. 
\end{eqnarray}

In what follows, we compute a lower-bound for the above fraction. Starting with the numerator, we have:

\begin{eqnarray*}
\bar{\chi}^i - \bar{\chi}^j 
&=&
\frac{\sum_{z \neq i,j}\left( x_z^i - x_z^j\right) \left(\frac{V}{\sum_{q\neq z} x_z^q}\right)+ x_j^i\left(\frac{V}{\sum_{q\neq j} x_j^q}\right)}{n} \\\nonumber
&&
- \frac{x_i^j\left(\frac{V}{\sum_{q\neq i} x_i^q}\right)}{n}\\\nonumber
&\geq&
\frac{V}{n}\left(\frac{n-2}{(n-1)M} + \frac{1}{(n-1)M} - \frac{M}{(n-1)}\right)\\\nonumber
&=&
\frac{V}{n}\left(\frac{n-2+1-M^2}{(n-1)M}\right)\\\nonumber
&\geq&
\frac{V}{n(n-1)M}\\\nonumber
&\geq&
\frac{V}{n^2M}.\nonumber
\end{eqnarray*}

The first inequality follows from the facts that for every agent $z \neq i,j$, $x_z^i > x_z^j$ and $\forall i,j,\, x_i^j \in \{1,\dots, M\}$. The second inequality follows from the assumption that $M \leq \sqrt{n-2}$. Focusing on the denominator of the fraction in~(\ref{proof-fair}), since $\forall q\in N, \zeta_q$ is the average of $n-1$ results from the BTS method, then the difference between $\zeta_j$ and $\zeta_i$ is always less than or equal to the difference between the highest and the lowest scores that can be returned by the BTS method (Equation~\ref{BTS}), which is equal to $3\ln\left(\frac{M}{\epsilon}\right)$ according to Lemma 1. Thus, we conclude that if:

\begin{displaymath}
\alpha \leq \frac{V}{3Mn^2\ln\left(\frac{M}{\epsilon}\right)},
\end{displaymath}

\noindent and $M \leq \sqrt{n-2}$, then the proposed mechanism is  fair.
\end{proof}

Intuitively, this proposition  means that the proposed mechanism can be made fair by reducing the influence of the truth-telling scores on agents' shares, so that these shares will depend almost entirely on the reported evaluations.

\begin{proposition}
If $\alpha \leq \frac{V}{2Mn\ln\left(\frac{M}{\epsilon}\right)}$, then the mechanism is individually rational.
\end{proposition}

\begin{proof}

We start the proof by observing that  $\forall i \in N, \,\bar{\chi}^i \geq 0$ (Equation~\ref{scaled-eval2}). Consequently, if agents' scores are positive, then their shares will also be positive. So, we restrict ourselves to the scenario where truth-telling scores are negative. Thus, for every agent $i \in N$, the following inequality must be true when $\zeta_i < 0$:

\begin{equation}
\label{ind-rat-proof}
\bar{\chi}^i + \alpha\, \zeta_i \geq 0 \equiv
\frac{\bar{\chi}^i}{-\zeta_i} \geq \alpha.
\end{equation}

In what follows, we compute a lower-bound for the fraction in~(\ref{ind-rat-proof}). Starting with the numerator, we have:

\begin{eqnarray*}
\bar{\chi}^i
&=&
\frac{\sum_{j\neq i} x_j^i\left(\frac{V}{\sum_{q\neq j} x_j^q}\right)}{n} \\ \nonumber
&\geq& 
\frac{\sum_{j\neq i} x_j^i\left(\frac{V}{M\,(n-1)}\right)}{n} \\ \nonumber
&\geq&
 \frac{V (n-1)}{M\,n(n-1)} \\\nonumber
\end{eqnarray*}

The inequalities follow from the fact $\forall i,j,\, x_i^j \in \{1, \dots, M\}$. Focusing on the denominator of the fraction in~(\ref{ind-rat-proof}), since $\zeta_i$ is the average of $n-1$ results from the BTS method, we can restrict ourselves to find the lowest negative score that can be returned by the BTS method. From Lemma 1, we know that this value is $-2\ln\left(\frac{M}{\epsilon}\right)$. Thus, we conclude that if:

\begin{displaymath}
\alpha \leq \frac{V}{2Mn\ln\left(\frac{M}{\epsilon}\right)}, 
\end{displaymath}

\noindent then the proposed mechanism is individually rational.
\end{proof}

Since agents' scores can be negative, the above proposition means that the resulting shares can always be positive, regardless the reported evaluations and predictions, if we reduce the influence of these scores on agents' shares.

\section{Empirical Evaluation}

In this section, we report an empirical investigation of the influence of the model and mechanism's parameters on agents' shares. In all experiments reported here, agents' truthful evaluations are drawn from the probability distribution of the random variable $\mathbb{H} = \lceil M \mathbb{B} \rceil$, where $\mathbb{B}$ is Beta-distributed with parameters $\alpha = \beta = 0.5$, \textit{i.e.}, $\mathbb{B}$ has a symmetric, U-shaped distribution. For creating a random prediction, we use the empirical distribution of $n-1$ evaluations drawn from the probability distribution of $\mathbb{H}$. Thus, the experiments reflect scenarios where most of the agents have extreme opinions. Lastly, agents always report their opinions truthfully.

\subsection{Parameter $M$}

The parameter $M$ defines the range of possible evaluations that an agent can give or receive. To better understand the influence of different values of $M$ on agents' shares, we performed the following experiment. We shared the reward $V = 1000$ among  $100$ agents using the proposed mechanism and the following values for $M$: $2, 5, 7, 10, 25, 50, 75, 100$. We used the parameters $\alpha = 10$ and $\epsilon = 10^{-4}$, and we observed the mean and the standard deviation of the resulting shares for different values of $M$. Figure 1 shows the results.

\begin{figure}
\centering
\includegraphics[width=\linewidth]{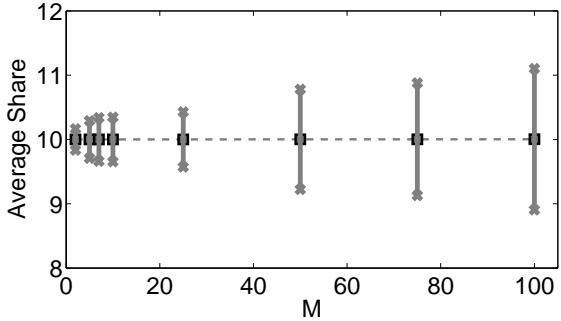}
\caption{Results of the experiment with different values for $M$. Average shares are represented by black squares, and standard deviations by gray lines. The dotted line is used to facilitate visualization.}
\vskip -10pt
\end{figure}

As can be seen in Figure 1, as $M$ increases, the standard deviation of the resulting shares also increases. Intuitively, this happens because the reported evaluations become more fine-grained, in that small differences between agents are recognized and specified by their peers, thus resulting in more diverse shares. It is important to note that this increased expressivity may be burdensome for the agents since they will have more possibilities to evaluate their peers, thus making the evaluation process more challenging. We argue that the underlying application may help to determine appropriate settings for $M$. Since the mechanism is budget-balanced and we used a fairly large population in this experiment, the average share stayed constant for different values of $M$.

\subsection{Parameter $\alpha$}

The parameter $\alpha$ of the proposed mechanism fine-tunes the weight given to the truth-telling scores. To better understand its influence on agents' shares, we performed the following experiment. We shared the reward $V = 1000$ among $100$ agents using the parameters $M = 10$, $\epsilon = 10^{-4}$, and $\alpha \in \{0.1, 1, 5, 10, 25, 50, 100, 500\}$. We ran this experiment $100$ times. We observed the total number of unfair shares and the total number of negative shares returned by the mechanism for different values of $\alpha$. An agent's share is considered unfair if that agent unanimously receives better evaluations than a peer, but its share is smaller than the peer's share. Thus, a mechanism is fair if it does not return unfair shares (see Definition 3). To compute the number of unfair shares, we made a pairwise comparison in each simulation step in which each returned share was compared to each other for determining whether the former was unfair or not. Table~\ref{tab:exp2} presents the results of this experiment.

\begin{table}[b]
\vskip -10pt
\centering
\caption{\label{tab:exp2} Results of the experiment with different values for $\alpha$.}
\begin{tabular}{|c|c|c|} \hline
$\alpha$ & \textbf{Unfair shares} & \textbf{Negative shares} \\\hline
0.1 &  0 & 0\\\hline
1   &  0 & 0\\\hline
5   &  0 & 0\\\hline
10  &  0 & 0\\\hline
25  &  0 & 0\\\hline
50  &  0 & 0\\\hline
100 &  0 & 8\\\hline
500 &  0 & 2543\\\hline
\end{tabular}
\end{table}

According to Proposition 4 and 5, we need to set $\alpha < 2.9\times 10^{-4}$ to mathematically ensure that the mechanism will be fair, and $\alpha < 0.044$ to mathematically ensure that the returned shares will always be greater than or equal to zero. From Table~\ref{tab:exp2}, we note that even with much higher values for $\alpha$, the mechanism did not return a single unfair share in this experiment. Further, the mechanism did not return a single negative share for $\alpha \leq 50$. This discrepancy between experiment and theory can be ascribed to the fact that the bounds for $\alpha$ are calculated based on worst-case scenarios, which are very unlikely to happen in practical applications. This implies that it is possible to promote truthfulness by using high values for $\alpha$ and still be able to obtain individual rationality and fairness.

\subsection{Parameter $n$}

The most stringent assumption made in this work is that the population of agents is large. This assumption is necessary for the proposed mechanism to be able to use the BTS method. We performed an experiment to investigate how this mechanism behaves when dealing with populations of different sizes. In detail, we studied how the size of the population affects the budget of the mechanism. We shared the reward $V = 1000$ using the parameters $M = 10$, $\alpha = 10$, $\epsilon = 10^{-4}$, and $n \in \{5, 10, 25, 50, 100, 150\}$. We executed the experiment 100 times. At the end of each simulation step, we computed the sum of the returned shares for each value of $n$. At the end of the experiment, we computed the averages and the standard deviations of these sums. Figure 2 shows the results.

\begin{figure}
\centering
\includegraphics[width=\linewidth]{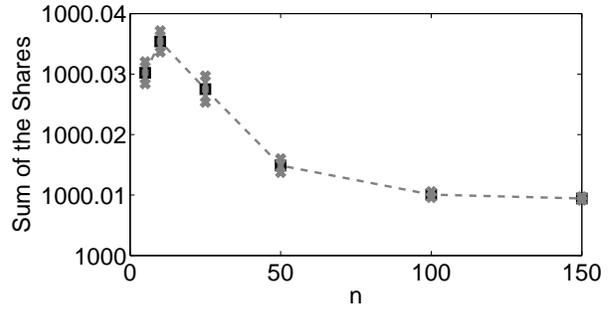}
\caption{Results of the experiment with different values for $n$. Black squares represent averages of the sum of the shares, and gray lines represent standard deviations. The dotted line is used to facilitate visualization.}
\vskip -10pt
\end{figure}

As can be seen in Figure 2, the mechanism loses more when $n \leq M$. Intuitively, since there are few agents to endorse a larger number of possible evaluations, the reported evaluations are very often surprisingly common. This implies higher truth-telling scores for the agents and, consequently, greater shares. Alternatively, agents' scores are more balanced when $n > M$. Since there are more agents than evaluations to be endorsed, the reported evaluations are not very often surprisingly common. Consequently, the average truth-telling score is not so high, and the mechanism's loss gradually decreases. An ANOVA test confirms that $n$ does indeed influence the resulting shares ($\rho < 0.0001$). The standard deviation of the sum of the shares also decreases when $n$ increases, thus supporting our claim that the scores are more balanced.

In conclusion, we note that a possible way to circumvent the assumption of a large population is to reduce the number of possible evaluations, \textit{i.e.}, to reduce the value of the parameter $M$. In this way, the influence of a single agent on the empirical distributions of evaluations may be reduced since these distributions will probably (but not necessarily) be more balanced. We suggest that a good rule of thumb is to use a value for the parameter $M \leq \sqrt{n-2}$, because at this point the number of different evaluations seems to be sufficiently smaller than the number of agents. Also, a value for $M$ satisfying this inequality helps to mathematically ensure fairness (Proposition 4). This rule has a strong empirical support in our experiment because the loss taken by  the mechanism is negligible when the inequality is satisfied, \textit{i.e.}, for $n = 100$ and $n=150$.

\section{Related Work}

Fair division has long been studied in cooperative game theory. The \emph{Shapley value}~\cite{osborne-rubinstein} is a key concept used in this field to distribute a joint surplus (or cost) among a set of agents. Roughly speaking, the Shapley value assigns a share to each agent equal to that agent's marginal contribution to the group. We note that sharing schemes based on marginal contributions, like the Shapley value, are not appropriate in our setting. The idea of marginal contribution is not objectively defined in our model because individual contributions are subjective information.

In the context of cooperative learning, Oakley \textit{et al.}~\cite{effective-teams} propose some guidelines to the effective design and management of teams of students. Slightly different from our model, each team member receives a common grade as the result of a joint academic work. These grades are adjusted through peer ratings (evaluations) to account for individual performance. In detail, a team grade is weighted by the average evaluation that a student receives to determine his or her final grade. A total of $9$ verbal evaluations are used, which are later converted to values inside the set $\{0, 12.5, 25, ...100\}$. Differently from our work, this rating scheme allows agents to make self-evaluations. Further, it does not promote truthfulness. Kaufman \textit{et al.}~\cite{bias} discuss the problems that may arise when using this rating system, $e.g.$, inflated self-evalua-tions and gender and racial bias. We believe that these problems may be even worse in our scenario because there is a joint reward to be shared, and not a common team grade.

Hence, the BTS method is an important component of our mechanism. We note that similar incentive-compatible methods which would require less information from the agents (\textit{i.e.}, only evaluations) could have been used (\textit{e.g.},~\cite{peer-prediction, Jurca2008a}).  However, most of these methods are not budget-balanced, which we believe is an important property in our setting.

\section{Conclusion}

In this paper, we proposed a game-theoretic model for sharing a joint, homogeneous reward based on the idea of subjective opinions. Each agent is asked to evaluate its peers as well as to predict how they will be evaluated. We introduced a mechanism to aggregate and use such opinions for determining agents' shares. The intuition behind the proposed mechanism is that each agent who believes that the others are telling the truth has its expected share maximized to the extent that it is well-evaluated and that it is truthfully reporting its opinions. Under the assumptions that agents are Bayesian decision-makers and that the underlying population of agents is sufficiently large, we showed that the proposed mechanism is incentive-compatible, budget-balanced, and tractable. We also presented strategies to make this mechanism individually rational and fair.

We implicitly assumed that agents are not participating in collusive agreements. However, there are many reasons why an agent may lie to benefit a peer. For example, in exchange for misreporting its evaluation, which may lead to a lower share for itself, a liar agent may receive a side-payment from the agent who benefits from the misreporting. Thus, an exciting direction for future research work is to study which kinds of collusive behavior may arise and how to avoid them.

\bibliographystyle{abbrv}
\bibliography{sigproc}  


\end{document}